\documentclass[12pt,reqno]{amsart}  

\usepackage{amsmath,amsthm,amsfonts,amssymb}

\theoremstyle{plain}
\newtheorem{thm}{Theorem}[section]
\newtheorem{prop}[thm]{Proposition}
\newtheorem{lem}[thm]{Lemma}

\theoremstyle{definition}
\newtheorem{definition}[equation]{Definition}

\numberwithin{equation}{section}

\newcommand{\id}{\mathrm{id}}
\newcommand{\diag}{\mathrm{diag}}

\newcommand\bb[1]{{\text{\bf#1}}}
\newcommand\bbz{\mathbb{Z}} 
\newcommand\bbr{\mathbb{R}} 
\newcommand\bbp{\mathbb{P}}
\newcommand\bbi{\bb{I}}
\newcommand{\scp}[1]{\langle #1\rangle}

\begin{document}
\baselineskip=16pt
\title[Monopoles on Sasakian three-folds] {Monopoles on Sasakian three-folds}

 \author[I. Biswas]{Indranil Biswas}

\address{School of Mathematics, Tata Institute of Fundamental Research,
Homi Bhabha Road, Bombay 400005, India}

\email{indranil@math.tifr.res.in}

\author[J. Hurtubise]{Jacques Hurtubise}

\address{Department of Mathematics, McGill University, Burnside
Hall, 805 Sherbrooke St. W., Montreal, Que. H3A 2K6, Canada}

\email{jacques.hurtubise@mcgill.ca}



\date{}

\begin{abstract}
We consider monopoles with singularities of Dirac type on 
quasiregular Sasakian three-folds fibering 
over a compact Riemann surface $\Sigma$, for example the Hopf fibration
$S^3\longrightarrow S^2$. We show that these correspond to holomorphic
objects on $\Sigma$, which we call twisted bundle triples. These are somewhat
similar to Murray's bundle gerbes. A spectral curve construction allows us
to classify these structures, and, conjecturally, monopoles.
\end{abstract}

\maketitle

\section{Introduction}

Since being introduced in the 1930's by Dirac, monopoles on a three-fold have occupied a 
place of privilege in the understanding of gauge theory. Dirac's monopoles were 
singular, defined over $\bbr^3$, and attached to the gauge group $U(1)$. In the 1970's, 
it was realized that introducing a non-Abelian gauge group allowed one to consider, on 
$\bbr^3$ at least, non-singular solutions. Like their close cousins,
namely instantons on 
$\bbr^4$, monopoles on $\bbr^3$ allow a holomorphic interpretation, and this was used to 
great effect by several authors, for example Ward \cite{Ward} and most notably Hitchin 
\cite{hitchinGeodesics} in constructing solutions. In parallel, work of Nahm 
\cite{Nahm}, and then Hitchin \cite{Hitchin-construction}, tied this complex 
interpretation to the Nahm transform, giving a very effective dictionary which allowed 
the classification of monopoles in 1983 by Donaldson \cite{DonaldsonSU2}. The work, 
originally done for the gauge group ${\rm SU}(2)$, was extended to classical gauge groups by 
Murray and Hurtubise in a series of papers \cite{Murray, HuMu, hurtubiseClassification}, 
and then by Jarvis to arbitrary reductive groups \cite{JarvisEuclideanMonopoles}.

Of course, one is not tied to $\bbr^3$, and one of the early extensions was to 
hyperbolic space; this case was studied in a beautiful paper by Atiyah 
\cite{Atiyah-monopoles-hyperbolic}. One can show, however, that non-singular and 
non-trivial monopoles cannot exist unless the space has a suitably large infinity. In 
particular one cannot have them on a compact manifold. Thus, in the latter case, one is 
led to admitting some singularities, and those which first appeared in the work of 
Dirac, and their analogues for general gauge groups, seem to be the most appropriate.

It was realized quite early on that the Dirac-type singularity leads to some most 
interesting geometry. Indeed, Kronheimer, in his Oxford MSc thesis, \cite{kronheimerMSc}, 
showed that the geometry of these Dirac monopoles is tied intimately to that of the Hopf 
fibration, and that one can define a lift of the singular monopole to a nontrivial 
fibration which smooths out the singularity. Pauly expanded and developed this idea in 
\cite{pauly}. Meanwhile, the singular monopoles turned up in a variety of contexts, 
linked for example to Nahm transforms of smooth configurations by Charbonneau 
\cite{benoitpaper}, and to gravitational instantons by Cherkis-Kapustin 
\cite{cherkis1999}. Most spectacularly, they mediate Hecke transforms, reinterpreted as 
a scattering by the monopole, and are an important ingredient in Witten-Kapustin's gauge 
theoretic interpretation of the geometric Langlands correspondence 
\cite{geometricLanglands}.

In this last interpretation, one is looking at singular monopoles on the product of a 
Riemann surface and an interval in $\mathbb R$; this was examined by Norbury 
\cite{Norbury-monopoles-boundary}. Charbonneau and Hurtubise, in \cite{ChHu}, then took 
up the case of self-Hecke transformations --- monopoles on the product of a Riemann 
surface and a circle. They proved a Kobayashi-Hitchin type correspondence for the 
singular monopoles, showing that they correspond to holomorphic vector bundles on a 
Riemann surface equipped with a meromorphic automorphism, thought of as the self-Hecke 
correspondence. These Kobayashi-Hitchin correspondences are a recurrent theme in gauge 
theory over K\"ahler manifolds, linking gauge theoretic solutions to certain field 
equations (the Hermite-Einstein condition) to holomorphic objects, allowing us, for 
example, to classify them. This has been developed most notably by Donaldson 
\cite{Donaldson-surfaces}, Uhlenbeck-Yau \cite{UY}, and Simpson 
\cite{Simpson-Hodge-structures} (see also \cite{Lubke-Teleman}).

Of course, there is no reason to restrict one's attention to the trivial line bundle 
over a Riemann surface, and the subject of this paper is to see what happens  over a more general
circle bundle $X$, which we will take to be positive. It turns out that the relevant 
geometry for our circle bundle is Sasakian geometry. A Sasakian structure exists on an 
arbitrary line bundle of positive degree on a compact Riemann surface. The relevant 
structure on the four-fold $X\times S^1$, instead of a K\"ahler structure, will be a 
Gauduchon metric. We mention that a simple example is the round three-sphere, fibering over the 
two-sphere; the four-fold is the Hopf surface. The Kobayashi-Hitchin paradigm 
extends to this situation, thanks to work of Buchdahl \cite{Buchdahl}. 
Unfortunately, his results only apply to the non-singular case. In the case studied 
by Charbonneau and Hurtubise \cite{ChHu}, the four-fold is K\"ahler, and the result of Simpson, 
which allows singularities, enables one to conclude that the
Kobayashi-Hitchin type correspondence is 
bijective. The corresponding generalization with singularities of Buchdahl's theorem 
remains unproven, though the full generalization of Simpson's results to the 
non-singular case with a Gauduchon metric in any dimension has recently been given 
by Jacob \cite{Jacob}.

Nevertheless, we can show that there are quite interesting holomorphic objects, of a 
fairly novel type, attached to the gauge fields, and we can show that the correspondence 
is injective. These objects can either be thought of as living on the three-fold $X$ 
(where one must give a suitable definition of holomorphic objects) or on the base Riemann 
surface $\Sigma$. In the latter context, they give objects which are rather reminiscent 
of Murray's bundle gerbes \cite{Murray2}.

Section 2 of this paper is devoted to recalling the necessary Sasakian and Gauduchon 
geometry on our circle bundle $X$. Section 3 considers monopoles on this three-fold, 
and defines the holomorphic objects on the three-fold which correspond to it. Section 4 
is quite brief and discusses the Kobayashi-Hitchin correspondence. Section 5 discusses 
the links between the holomorphic objects we have defined, and their reductions to the 
Riemann surface. Section 6 discusses the more general case of a circle bundle over an 
orbifold Riemann surface.

\section{Sasakian geometry}

\subsection{Quasiregular Sasakian manifolds}

Let $X$ be a compact quasiregular Sasakian three-fold, with metric $g$. These are 
manifolds with a contact structure and a metric, compatible in the sense that there is a 
unit (Reeb) vector field $\xi$ orthogonal to the contact planes, which acts on the 
manifold as a Killing field. (A useful reference is the book of Boyer and Galicki
\cite{BG}.) The orbits of $\xi$ are compact; under this hypothesis the manifold has a circle 
action by isometries, and the quotient
\begin{equation}\label{e1}
\pi\,:\,X\,\longrightarrow\, \Sigma
\end{equation}
by the flow is a compact orbifold Riemann surface equipped with a K\"ahler
form $\omega$. The K\"ahler
structure $\omega$ is uniquely determined by the condition
that $\pi$ is a Riemannian submersion with respect to $g$ and $\omega$.
The quasiregular Sasakian three-fold $(X\, ,g\, ,\xi)$ is called {\it regular}
if the circle action on $X$ is free. So a regular Sasakian three-fold is a principal
$S^1$-bundle over a compact Riemann surface $\Sigma$ equipped with a
K\"ahler form $\omega$. From now on, until the end of Section \ref{se5}, we specialize to the regular
Sasakian manifolds.

Let $\alpha$ be the normalized contact form on $X$, so 
$$\alpha(\xi) \,=\, 1\, ,\ ~\alpha|_{\xi^\perp}\,=\, 0\, ,$$
where $\xi$ is the above Killing field. If, in addition, we take, locally,  one-forms $dz,
d\overline z$ on $\Sigma$, and pull them back to $X$ using $\pi$ in \eqref{e1}, one has a
local basis of complex $1$-forms 
 $$\alpha\, , dz\, , d\overline z$$ on $X$. Let 
 $$\xi\, , v_z\, , v_{\overline z}$$
be the dual basis of vector fields. While $\xi$ is the Reeb vector field, both $v_z\, ,
v_{\overline z}$ are orthogonal to $\xi$, as well as being mutually orthogonal.
The contact property tells us that the 2-form $d\alpha$ is non-degenerate.
In addition, it is a lift from the Riemann surface. The form on $X$ which pulls back
to $d\alpha$ will be denoted by $\omega$. One has a basis 
 $$d\alpha\,=\, \pi^*\omega \,=\,  \sqrt{-1} \mu(z,\overline{z})dz\wedge d\overline z\,,~\
\alpha\wedge dz\, ,~\ \alpha\wedge d\overline{z}$$
for the complex two-forms $X$, where $\omega$ is the above K\"ahler form
on $\Sigma$. The Lie brackets of the vector fields are
  $$[v_z,v_{\overline{z}}] \,=\, -\sqrt{-1}\mu(z,\overline{z}) \xi\, ,\quad
[\xi,v_z]\,=\, [\xi,v_{\overline{z}}]\,=\, 0\, .$$
These follow using the equation 
$\alpha([v_z,v_{\overline{z}}])\,=\,-d\alpha(v_z,v_{\overline{z}})$ and the fact that 
the vector field $\xi$ is Killing. In these bases, one can give the metric by
$$ g\,=\, \pi^*h + \alpha\otimes \alpha$$
with $h$ being the Hermitian structure on $\Sigma$ associated to $\omega$.
Finally, one has a volume form $d\alpha\wedge\alpha$ on the three-fold $X$.

\subsection{Sasakian geometry and K\"ahler geometry}

One defining property of a Sasakian three-fold is that on the cone $M\,=\, \bbr^+\times X$
over $X$, the metric $dr^2 + r^2g$ is K\"ahler, where $g$ is the Riemannian metric on
$X$. The K\"ahler form is
\begin{equation}\label{f2}
\Omega \,=\, r^2d\alpha -2r\alpha\wedge dr \,=\, r^2 \pi^*\omega -2r\alpha\wedge dr\, .
\end{equation}

On the cone $M$, there is a basis of vector fields $\frac{\xi}{r}\, ,
\frac{\partial}{\partial 
r}\, , \frac{v_z}{r}\, , \frac{v_{\overline z}}{r}$. They have constant norm in $r$,
and the first two are mutually orthogonal while being orthogonal to the others.
 
\subsubsection{Complex structures}

We have the complex structure
on $X$ inherited from the Riemann surface $\Sigma$, to which
one adds $J(\frac{\partial}{\partial r}) \,=\, \frac{\xi}{r}$, and so  
$$J(\frac{\xi}{r} + \sqrt{-1}  \frac{\partial}{\partial r})\,=\, \sqrt{-1}
(\frac{\xi}{r}+\sqrt{-1}\frac{\partial}{\partial r})\, ,\ J(\frac{v_z}{r}) \,=\,
\sqrt{-1} \frac{v_z}{r}\, ,$$ spanning the $(1,0)$
part of the complexified tangent space, dually, $$J(r\alpha-\sqrt{-1}  {dr} )\,=\,
\sqrt{-1}(r\alpha-\sqrt{-1}  {dr} )\, ,\ J(rdz) \,=\,\sqrt{-1}(rdz)\, ,$$  spanning the
$(1,0)$ forms, and
$$r\alpha+\sqrt{-1} {dr}\,  , \ rd\overline z$$
for the $(0,1)$ forms. The real subspace of the $(1,1)$ forms is spanned by
$$r\alpha\wedge  {dr}\, ,\ r^2d\alpha\,=\, \sqrt{-1}\mu r^2 dz\wedge d\overline{z}
\, ,\ \sigma_3 \,=\, Re((r\alpha-\sqrt{-1} dr)\wedge d\overline{z})\, ,
$$
$$
\sigma_4
\,=\, Re((r\alpha+\sqrt{-1} dr )\wedge dz)\, .$$

We have the volume form on $M$
$$\Omega\wedge \Omega \,=\, -4r^3 d\alpha\wedge \alpha \wedge dr\, ,$$ 
where $\Omega$ is constructed in \eqref{f2}.

\subsection{Sasakian geometry and Gauduchon geometry} 

A pointwise positive $(1\, ,1)$--form $\zeta$ on a complex surface is called
\textit{Gauduchon} if $\partial\overline{\partial}\zeta\,=\, 0$.

Now let us consider instead of $\Omega$, the form on $M$
\begin{equation}\label{f1}
\widetilde \Omega \,=\, \frac{1}{r^2}\Omega \,=\, d\alpha -2 \alpha \wedge \frac{dr}{r}
\,=\, d\alpha -2 \alpha \wedge dt\, ,
\end{equation}
setting $t\,=\, \log(r)$.

\begin{lem}\label{lem1}
The following holds:
$$\partial\overline{\partial}(\frac{1}{r^2}) \,=\, -\frac{\sqrt{-1}}{r^2}
(d\alpha +2\alpha\wedge dt)\, .$$
As a consequence, $$\partial\overline{\partial}(\frac{1}{r^2}\Omega)\,=\,
\partial\overline{\partial}(\frac{1}{r^2})\wedge \Omega \,=\, 0\, ,$$
so that the form $\widetilde\Omega$ in \eqref{f1} is Gauduchon.
\end{lem}

The form $\widetilde \Omega$ in \eqref{f1} is invariant under the flow of the Reeb vector field, and
in addition is also invariant in the additive time ($t$-) direction. We note that Lemma \ref{lem1} shows
that there is a time invariant Gauduchon metric $\widetilde \Omega$ on the manifold
\begin{equation}\label{N}
N\,=\, X\times S^1\, .
\end{equation}
It is this compact Gauduchon surface $(N\, , \widetilde\Omega)$ that will be used.

Before giving a few geometric properties, we first re-scale the bases given above on 
$X$. We had one-forms $dz, d\overline z$ on $\Sigma$, lifted to $X$, giving a local 
basis of forms $\alpha, dz, d\overline z$, and dually, vectors $\xi, v_z, v_{\overline 
z}$. On $N$ (defined in \eqref{N}), we use a basis ${\xi}\, , \frac{\partial}{\partial 
t}\, , v_z \, , v_{\overline z}$; they have constant norm in $t$ (with respect to 
$\widetilde \Omega$), with the first two being orthogonal and normal to the latter two
which are isotropic.

We have the complex structure $J(\frac{\partial}{\partial t}) \,= \, {\xi}$, and so  
$$ \xi + \sqrt{-1} \frac{\partial}{\partial t}\, , \quad {v_z} $$ span the $(1,0)$ part of the
complexified tangent space. Dually,
$$ \alpha-\sqrt{-1}  {dt}\,  , ~\ dz$$ span the $(1,0)$ forms, while
$$ \alpha+\sqrt{-1} {dt}\,  , ~\ d\overline z$$
span the $(0,1)$ forms.  We have a real basis of $(1,1)$ forms $$d\alpha\, , \alpha\wedge dt\, ,
\widetilde{v}_3\,=\, \frac{v_3}{r^2}\, ,  \widetilde{v}_4 \,=\, \frac{v_4}{r^2}\, .$$

Let $\widetilde L$ denote the exterior product operation of forms by
$\widetilde \Omega$, and let
$\widetilde \Lambda$ denote the adjoint of $\widetilde L$. If $\eta$ is a $(p,q)$ form,
then $[\widetilde L\, , \widetilde\Lambda]
(\eta)\,=\, (p+q-2)(\eta)$. Thus, for a 2-form $\eta$ on $M$, one should have the component
$(\widetilde L)^2\widetilde\Lambda(\eta)\,=\, \widetilde L\widetilde\Lambda \widetilde L(\eta)
\,= \,2\widetilde L(\eta)$, and so 
$$
\widetilde \Lambda(\eta)\widetilde  \Omega\wedge \widetilde \Omega\,=\,
2\eta\wedge \widetilde \Omega\, .
$$

Let us compute the Laplacian over $X$. The coframe $\alpha\, , dz\, , d\overline z$ satisfies
the condition that $\alpha$ is orthogonal to the other two, and $dz\, , d\overline z$ are
isotropic, with $g(\alpha, \alpha) \,=\,  \frac{1}{2}$ and $g(dz, d\overline z)\,=\,\mu^{-1}$.
The volume form is $\sqrt{-1} \mu\ \alpha\wedge dz\wedge d\overline z\,=\,
d\alpha\wedge\alpha$. From the relation $g(a,b) d\,\text{vol} \,=\, a\wedge *b$, one has 
$$*\alpha\,=\, \frac{1}{2} d\alpha\, ,~ \  *dz\,=\, \sqrt{-1}  \alpha\wedge dz\, 
,~\  *d\overline z \,=\, - \sqrt{-1}\alpha\wedge d\overline z\, .
$$
Hence, for a function $f$ on $X$
$$ *df \,=\, *\left(\xi(f)\ \alpha + v_z(f)\ dz + v_{\overline{z}}(f) d\overline z\right)
\,=\, \xi(f)\  \frac{{d\alpha} }{2} + v_z(f)\  \sqrt{-1} \alpha\wedge dz -
v_{\overline{z}}(f) \sqrt{-1}\alpha\wedge d\overline z\, ,$$
$$ d*df \,=\,  \xi^2(f)\   \frac{ \alpha\wedge d\alpha}{2}  + v_{\overline{z}}(v_z(f))\ 
\sqrt{-1} \alpha\wedge dz\wedge d\overline z + v_z(v_{\overline{z}}(f) )\sqrt{-1}
\alpha\wedge dz\wedge d{\overline z}\, ,$$
$$ \Delta(f) \,=\, *d*df \,=\, 
\frac{1}{2} \xi^2(f)\ + \   \mu ^{-1}( v_{\overline{z}}(v_z(f))+
v_z(v_{\overline{z}}(f) ))\, .$$
Similarly, given a vector bundle $E$ on $X$ equipped with a connection $\nabla$,
one can extend $\nabla$ to an operator $$\nabla\,:\,\Gamma(E\otimes \wedge^k(X))
\,\longrightarrow \,\Gamma(E\otimes \wedge^{k+1}(X))\, ,$$ and one has on sections of $E$:
$$\Delta (s) \,=\, *\nabla*\nabla(f)\,= \,\frac{1}{2} (\nabla_\xi)^2(s) \ + \  
\mu^{-1}( \nabla_{v_{\overline{z}}}(\nabla_{v_z}(s))+ \nabla_{v_z}
(\nabla_{v_{\overline{z}}}(s) ))\, ,$$
with, on $L^2$ norms, $\langle s\, ,\Delta(s)\rangle\,=\, - \langle \nabla(s)\,
, \nabla(s)\rangle$.

\section{Bundles, Hermite-Einstein monopoles and holomorphic structures}

\subsection{The Hermite-Einstein condition}

Now assume that we have a Hermitian vector bundle $E$ over $N$, equipped
with a Hermitian Chern connection, which we write as:
$$\nabla_{v_z}  dz+ \nabla_{v_{\overline z}}  d\overline{z} +
\nabla_ {\xi} \alpha + \nabla_{\frac{\partial}{\partial t}} dt
$$
$$
=\, ( {v_z}  +A_{ v_z} ) dz + ( v_{\overline{z}}  +
A_{ v_{\overline{z}}} ) d{\overline{z} }+ ( \xi +A_{ \xi} )  \alpha +
(\frac{\partial}{\partial t}+ \phi)dt\, .$$
Recall that the Lie brackets of our vector fields on $X$ are of
the form $$[v_z,v_{\overline{z}}]
\,=\, -\sqrt{-1}\mu(z,\overline{z}) \xi\, ,\  [\xi,v_z]
\,= \,[\xi,v_{\overline{z}}]\,=\, 0\, ,$$ while on $N$, we have
$[\frac{\partial}{\partial t}, \xi ]= 0$. The curvature tensor is:
\begin{align*}F\,=\, &\ (( \sqrt{-1}\mu(z,\overline{z}))^{-1}[\nabla_{ v_z } ,
\nabla_{ v_{\overline{z}}} ] +   \nabla_{ \xi})\ d\alpha +( [\nabla_{ \xi} 
,\nabla_{\frac{\partial}{\partial t}}] ) \   \alpha \wedge dt\\ &
+[ \nabla_{ \xi} ,\nabla_ {v_z} ]\  \alpha\wedge dz + [\nabla_ {\xi} ,
\nabla_ {v_{\overline{z}}}]\  \alpha\wedge d\overline{z} \\& +
([\nabla_{\frac{\partial}{\partial t}},\nabla_ {v_z} ] )\ dt\wedge dz +
([\nabla_{\frac{\partial}{\partial t}},\nabla_{ v_{\overline{z}}}] )\ 
dt\wedge d\overline{z} \\ =&\ F^{1,1} + F^{2,0} + F^{0,2}\\ 
\buildrel{{\mathrm def}}\over{=}\,& \   (F_\Sigma \  d\alpha + F_\alpha  \  
\alpha \wedge dt + F_3 \widetilde{v}_3 + F_4 \widetilde{v}_4) + F^{2,0} + F^{0,2}\, .
\end{align*}
When the connection is invariant in the $t$-direction, it can be put in a gauge for which 
the connection matrices are invariant in the $t$-direction. If $\phi$ is the $ 
dt$-component of the connection, the commutators $[\nabla_v, 
\nabla_{\frac{\partial}{\partial t}}]$ become $\nabla_v(\phi)$, in particular, 
$F_\alpha$ becomes $\nabla_\xi\phi$.
 
We fix once and for all a collection of points $P\,=\,\{p_1,\cdots ,p_\ell\}
\,\subset\, X$. Define $q_i\,:=\, \pi(p_i)\, \in\, \Sigma$. We also fix  sequences $\vec{k}_i
\,=\, (k_{i,1},\cdots ,k_{i,n})$ of integers associated with $\{p_i\}_{i=1}^\ell$, and order
the indices so that $k_{i,1}\,\geq\, \ldots \,\geq\, k_{i,n}$. 

\begin{definition}\label{HEmonopole} A {\it ${\rm U}(n)$-Hermite-Einstein monopole
with constant $C$, Dirac singularities of type $\vec{k}_i$ at $p_i$} will be a rank
$n$ Hermitian vector bundle $E$ on $X$, equipped with a Hermitian connection $\nabla$
and a skew Hermitian endomorphism (``Higgs field'') $\phi$, defined away from the points
$p_i$, such that
\begin{itemize}
\item When lifted to $N$ (with $\phi$ becoming the connection component along the extra
circle direction), the result is not only compatible with the metric but also the
complex structure, so that $F^{0,2}\,=\,F^{2,0}\,= 0$. 
The $F^{0,2}\,=\,0$ condition is
$$[\nabla_{ v_{\overline{z}}} , \nabla_{\xi} - \sqrt{-1}\phi]\,=\,0\, .$$
Taking complex conjugates, the $F^{2,0}\,=\,0$ condition is
$$[\nabla_ {v_z} , \nabla_{\xi} + \sqrt{-1}\phi]\,=\,0\, .$$

\item The lifted connection satisfy the Hermite-Einstein condition 
$$\widetilde \Lambda(F)\,= \,-\sqrt{-1}\ C\cdot\bbi\, .$$
More explicitly: 
\begin{equation}\label{eC}
F_\Sigma -  \frac{F_\alpha }{2} \,=\, \frac{2((\sqrt{-1}
\mu(z,\overline{z}))^{-1}[\nabla_{v_z} , \nabla_{v_{\overline{z}}}] +  \nabla_{ \xi})
- ( \nabla_\xi \phi) }{2} 
\,=\,   -\sqrt{-1}(C\cdot\bbi)\, .
\end{equation}

\item The singularities at $p_i$ are of {\it Dirac type}, as defined below.
\end{itemize}
\end{definition}
 
\begin{definition}\label{def:Dirac-type}
Let $Y$ be a three-manifold, equipped with a metric, and let
$p$ be a point of $Y$. Let $R$ denote the locally defined function
on $Y$ given by the geodesic distance to $p$. Let
$(t\, ,x\, ,y)$ be coordinates centered at $p$ with respect to which the metric is
of the form $(\id + O(R))$ as $R\,\to\,0$.  Let $\psi\, , \theta $ be, as
above, angular coordinates on the sphere $R\,=\,c$, so that $R\, ,\psi\, ,\theta$
provide standard spherical coordinates on a neighborhood $B^3$ of $p$ defined
by the inequality $R\,<\,c$. We say that a solution to the Hermite-Einstein monopole
equations $(E\, ,\nabla\, , \phi)$ on
$Y\setminus\{p\}$ has a \emph{singularity of Dirac type, with weight
$\vec{k}\,=\,(k_1\, ,\cdots\, ,k_n)$ at $p$} if 
\begin{itemize}
\item there is a unitary isomorphism $I$ of the restriction of the bundle $E$ to
$B^3\setminus\{p\}$ with a direct sum of line
bundles $L_{k_1}\oplus\cdots\oplus L_{k_n}$, where $L_k$ is the pullback from $S^2$ of
the standard line bundle of degree $k$, and 

\item under the isomorphism $I$, in the trivializations of $E$ over the  two open subsets
$\theta\,\neq\, 0$ and $\theta\,\neq\,\pi$ of $B^3$ induced by standard trivializations of
the line bundles $L_{k_i}$, so that the $E$-trivializations have transition function
$\diag(e^{\sqrt{-1}k_1\psi}\, ,\cdots\, ,e^{\sqrt{-1}k_n\psi})$, one has,
in both trivializations,
$$
\phi \,=\, \frac {\sqrt{-1}}{2R}\diag(k_1,\cdots,k_n) + O(1)\, ,\ ~ \nabla (R \phi)\,=\,
O(1)\, .
$$
\end{itemize}
\end{definition}

\subsection{Holomorphic structures}

The aim is to highlight a 
Kobayashi-Hitchin correspondence for our monopoles: they should yield some holomorphic 
objects which classify them. Obviously, as $X$ is three-dimensional, this holomorphic 
data must either be linked to the complex curve $\Sigma$, or to the complex surface $N$. 
In the end, we will do both; we begin by saying what our holomorphic objects on $N$ 
become once one restricts them to $X$.

Let $E$ be a complex $\mathcal C^\infty$ vector bundle on $X$. Let $T\, \subset\,
TX$ be the orthogonal complement to the vector field $\xi$. We note
that $T$ being isomorphic under 
projection $d\pi$ to $T\Sigma$, has a complexification which splits as $\widetilde 
T^{1,0}\Sigma\oplus \widetilde T^{0,1}\Sigma$. So $\widetilde
T^{1,0}\Sigma$ and $\widetilde T^{0,1}\Sigma$ are identified with $\pi^*T^{1,0}\Sigma$
and $\pi^*T^{0,1}\Sigma$ respectively.

\begin{definition}
A holomorphic structure on $E$ over $X$ (or, locally, on an open set of $X$) will be given
by specifying first order operators
$$\nabla^{0,1}_\Sigma\,:\, \Gamma(E)\,\longrightarrow\, \Gamma(E\otimes (
{\widetilde T}\Sigma^{0,1})^* )\, ,\ ~ \nabla^c_\xi\,:\, \Gamma(E)\,\longrightarrow \,
\Gamma(E))\, ,$$
which are locally of the form
$$s\,\longmapsto\, (v_{\overline z}(s) + A_\Sigma^{0,1}(s))d\overline z\, , \ ~
s \,\longmapsto \,\xi(s)+ \phi^c(s)\, ,$$
and which commute. 
\end{definition}

A holomorphic structure on $E$ over $X$ is a reduction of a holomorphic structure 
for a bundle on $N\,=\, S^1\times X$, corresponding to an integrable $\overline 
\partial$ operator over $N$, invariant in the circle direction on $N\,= \,S^1\times 
X$.

We note that given any open subset $U\,\subset\,\Sigma$, and any section $\psi\,:\,U\, 
\longrightarrow\, X$ of $\pi$, the two operators $\nabla^{0,1}_\Sigma$ and $\nabla^c_\xi$ can 
be combined to give a $\overline \partial$ operator for $E$ on $\psi(U)$. We can 
think of the result as a holomorphic bundle $E_\psi$ over $U$. Given two such 
sections $ \psi\, , \tau$ on $U$, we can choose paths along the circle orbits from $\psi(U)$ 
to $\tau(U)$, and integrate $\nabla^c_\xi$ along these paths, from $\psi(U)$ to 
$\tau(U)$ to obtain a map
\begin{equation}\label{om}
\rho_{\psi\tau}\,:\, E_{\psi(U)}\,\longrightarrow\, E_{\tau(U)}\, .
\end{equation}
If these paths are chosen in a continuous fashion, this $\rho_{\psi\tau}$ will
be a holomorphic isomorphism; again it can be thought of as a vector bundle isomorphism 
$\rho_{\phi,\tau}\,:\,E_\psi\,\longrightarrow\, E_\tau$ over $U$. We note that there are choices 
involved in the definition of $\rho_{\phi,\tau}$, as to the direction along the 
circle orbits and more generally the winding number. If $\psi(U)$ and $\tau(U)$ do 
not intersect, we choose to go from $\psi(U)$ to $\tau(U)$, in the positive 
direction of the circle action, less than one full circle. If $\tau\,= \,\psi$, one can 
also choose one full positive turn around the circle, giving a monodromy $G_\psi\,:\, 
E_\psi\,\longrightarrow\, E_\psi $ over $U$.

\begin{definition}
A {\it meromorphic structure on $E$ with poles at $P\,= \,\{p_1\, ,\cdots\, ,p_\ell\}$} 
is first a holomorphic structure on $E$ over the complement $X\setminus P$. One asks in 
addition that the structure be meromorphic at each $p_i$ in the following sense. Let $U$ 
be an open subset of $\Sigma$ containing $q_i\,=\,\pi(p_i)$. For any pair of sections of 
$\pi$
$$\psi\, , \tau\,
\,:\, U\,\longrightarrow\, X$$ with disjoint images, one constructs $\rho_{\psi\tau}$
in \eqref{om}.
The result is a holomorphic isomorphism away from $q_i$, and also at $q_i$ if the 
paths of integration of $\nabla^c_\xi$ on the fibers above $q_i$ do not contain a 
$p_j$. We ask that the map, in more generality, be meromorphic at $q_i$ even if the 
paths of integration contain a $p_j$.

Let us choose sections $\psi, \tau$ of the projection $\pi\, :\,X\,\longrightarrow\, 
\Sigma$ with disjoint image on an open set $U$ containing $q_i$ such that the path along 
the circle orbit over $q_i$ from $\psi(q_i)$ to $\tau(q_i)$ passes through $p_i$ once in 
the positive direction, and not any other $p_j$. Let us also choose a coordinate $z$ on 
$\Sigma$ with $z\,=\,0$ corresponding to $q_i$. We say that the meromorphic structure 
has a {\it pole of type $\vec{k}_i \,=\,(k_{i,1}\, ,\cdots\, ,k_{i,n})$ at $p_i$} if the 
map $\rho_{\phi,\tau}\,:\, E_\psi\,\longrightarrow\, E_\phi$ is of the form
$$\rho_{\phi,\tau}\,=\, F(z) {\rm diag}(z^{k_{i,1}}, z^{k_{i,2}},\cdots ,
z^{k_{i,n}}) G(z)\, ,$$ 
with $F, G$ holomorphic and invertible. We note that the order $k_{i,1}$ of
the ``pole'' can be either positive or negative.
\end{definition}
 
For a monopole, the operators for the meromorphic structure are simply 
$\nabla_{v_{\overline z}}, \nabla^c_\xi =  \nabla_\xi -\sqrt{-1}\phi$. We have, as in \cite{ChHu}:

\begin{prop}
A ${\rm U}(n)$--Hermite--Einstein monopole with constant $C$, Dirac singularities of type
$\vec{k}_i$ at $p_i$ determines a meromorphic structure on $E$ with poles at $P$,
of type $\vec{k}_i$ at $p_i$.
\end{prop}

The integrability of the holomorphic structure away from the singularity follows from
the equation $F^{0,2}\,=\,0$ satisfied by an Hermite-Einstein monopole. The meromorphic
behavior near the
singularities follows from the fact that the singularities are of Dirac type, and is
proven in \cite[Proposition 2.5]{ChHu}. A local version of this structure, on the
three-sphere, was considered by Pauly in \cite{Pauly-spherical-monopoles}.

\subsection{A degree}

Surfaces equipped with a Gauduchon metric give a well defined numerical degree for a
holomorphic bundle, by integrating against the trace of the curvature of a Chern
connection. The Gauduchon condition ensures that the integral is independent
of the Hermitian structure on the bundle. It should be mentioned that unlike   
the K\"ahler case, the degree can move continuously in a family of vector bundles.
In particular, the degree is no longer a topological invariant.

In our case, as the data is invariant
along the time $t$ direction, we obtain an expression for
the degree of an $t$-invariant bundle $E$ with connection $\nabla$ and
curvature $F$ by integrating along $t\,=\,0$, i.e., on the manifold $X$:
\begin{align*}\deg(E) ~\, = ~\,&  \sqrt{-1}\ Vol(X)^{-1}\int_X i(
\frac{\partial}{\partial t})(tr (F)\wedge \widetilde\Omega)\\
 = ~\,& \frac{  \sqrt{-1}}{2}\ Vol(X)^{-1}\int_X \widetilde\Lambda (tr (F))
i(\frac{\partial}{\partial t})(\widetilde\Omega\wedge \widetilde\Omega)\\ =~\,
& \frac{  \sqrt{-1}}{2}\ Vol(X)^{-1} \int_X \widetilde\Lambda (tr (F))
(4 d\alpha\wedge \alpha)\, ,\end{align*}
where $i(\frac{\partial}{\partial t})$ denotes the constriction of forms using the
vector field $\frac{\partial}{\partial t}$. Pursuing further, if one has the decomposition
$$tr(F)^{1,1} \,=\, tr(F)_\Sigma ( d\alpha) + tr(F)_\alpha (\alpha\wedge dt) +
tr(F)_{3} \widetilde\sigma_3+ tr(F)_{4} \widetilde\sigma_4\, ,$$
then in view of the equality $F_\alpha \,=\, \nabla_\xi(\phi)$, the integral becomes 
\begin{equation}\label{g1}
\deg(E) \,=\, \sqrt{-1} \ Vol(X)^{-1}\int_X (2tr(F)_\Sigma -
tr(\nabla_\xi\phi )( d\alpha\wedge \alpha)\, .
\end{equation}
We remark that $tr(\nabla_\xi\phi )\,=\, \xi(tr(\phi))$; using the fact that
$\int_{S^1} \xi(tr(\phi)) \alpha \,=\,0$, the integral in \eqref{g1} is then
\begin{align} \deg(E) \,=\,   & 2\sqrt{-1} \ Vol(X)^{-1}\int_X ( tr(F)_\Sigma
d\alpha)\wedge \alpha)\\ =\, & 
    2\sqrt{-1} \ Vol(X)^{-1}\int_X tr(F) \wedge \alpha\, .\nonumber\end{align}
If the Hermite-Einstein equation is satisfied, the degree is $2nC$, where $n$ is the rank
and $C$ is the constant in \eqref{eC}.

\begin{definition}
A {\it meromorphic section} of a meromorphic structure on $E\,\longrightarrow\, X$ is a
$C^\infty$ section of $E$ over $X\setminus P$ lying in the kernel of the operators
$\nabla^{0,1}_\Sigma$ and $\nabla^c_\xi$.
\end{definition}

\begin{prop}\label{prop1}
Consider $C$ in \eqref{eC}. If $C\,<\,0$, an Hermite--Einstein monopole has no non-zero 
meromorphic sections. If $C\,=\,0$, the only possibility for a section $s$ is as a 
covariant constant section, lying in the kernel of $\phi$. In particular, $s$ then 
defines a rank one Hermite--Einstein monopole summand, so that the monopole splits as a 
direct sum of a rank one monopole and a rank $n-1$ monopole.
\end{prop} 
 
\begin{proof}
One has the identity
\begin{align*}
 \mu^{-1}(\nabla_{v_z}\nabla_{v_{\overline{z}}} )\ +& \frac{1}{4} (\nabla_{\xi}+
\sqrt{-1}\phi)(\nabla_{\xi}-\sqrt{-1}\phi)\\ =~\,&\
\frac { \mu^{-1}}{2}(\nabla_{v_z}\nabla_{v_{\overline{z}}}  +
\nabla_{v_{\overline{z}}}\nabla_{v_z}  ) + \frac{\sqrt{-1}}{2}F_\Sigma -\sqrt{-1}
\frac{\nabla_\xi}{2}  \\ &\quad + \frac{1}{4}(\nabla^2_{\xi} + \phi^2-
\sqrt{-1}F_\alpha)\\ =~\,& \frac{1}{2}\Delta + \frac{1}{4} \phi^2 + 
\frac{\sqrt{-1}}{4} (2F_\Sigma-F_\alpha) -\sqrt{-1}\frac{\nabla_\xi}{2}\, .
\end{align*}
In particular, applying $\frac{1}{2}\Delta + \frac{1}{4} \phi^2 +
\frac{\sqrt{-1}}{4} (2F_\Sigma-F_\alpha) -\sqrt{-1}\frac{\nabla_\xi}{2}$ to a holomorphic
section gives zero.
Now start with a holomorphic section $s$. We have 
\begin{align*}
\frac{ C }{2}|s|^2_{L^2}&~\,=~\,\int_X
\scp{s,\frac{\sqrt{-1}}{4} (2F_\Sigma-F_\alpha)s}d\,{\rm vol}\\
&\geq ~\,
  \int_{X}\bigl(\scp{s,\frac{\sqrt{-1}}{4} (2F_\Sigma-F_\alpha)s}
 -\frac{1}{2}|\nabla s|^2- \frac{1}{4}|\phi s|^2\bigr)d\, {\rm vol}\\
&=~\,\int_{X}
\scp{s,(\frac{1}{2}\Delta + \frac{1}{4} \phi^2 + \frac{\sqrt{-1}}{4}
(2F_\Sigma-F_\alpha))s}d\,{\rm vol}\\
&=~\,\int_{X}
\scp{s,(\frac{1}{2}\Delta + \frac{1}{4} \phi^2 + \frac{\sqrt{-1}}{4}
(2F_\Sigma-F_\alpha)-\sqrt{-1}\frac{\nabla_\xi}{2})s}d\,{\rm vol}\\
&=~\,0\, .
\end{align*}

The third step involves an integration by parts. One checks that this causes no 
difficulties at the singularities.  For the fourth, one has the fact that the 
integrals along the circles in $X$ of $\xi\langle s\, ,s\rangle \,=\,
2\langle s\, , \nabla_\xi(s)\rangle$ is zero. 
Thus, unless $C$ is positive or zero, one finds $s\,=\,0$. If $C$ is zero, then we have 
$\nabla(s) \,=\, \phi(s) \,=\, 0$. This then tells us that the orthogonal complement of 
$s$ is also an invariant summand under the connection, and that it also is 
invariant under $\phi$.
\end{proof}

Proposition \ref{prop1} tells us in effect that our notion of degree gives an 
appropriate definition of stability. Indeed, if we have a vector bundle $E$ on $N$ with 
a ``$\overline\partial$''- operator $(\nabla_{v_{\overline 
z}}\, ,\nabla_\xi-\sqrt{-1}\phi)$ which is integrable, one can extend it as the 
Chern connection, by specifying a Hermitian metric. We saw above that one then 
has a well defined degree, independent of the choice. One can define meromorphic 
sub-bundles as bundles invariant under $(\nabla_{v_{\overline 
z}}\, ,\nabla_\xi-\sqrt{-1}\phi)$. Define the degree of
meromorphic sub-bundles in the same way. Define 
the slope $\mu(F)$ of any nonzero sub-bundle $F$ in the usual way as the quotient
of the degree by the rank.

\begin{definition} We say that the bundle $E$ on $N$ is stable (respectively, semistable) if
for all holomorphic sub-bundles $0\,\not=\, F\, \subsetneq\,E$ invariant under translation by $t$, 
$$\mu(F)\,<\, \mu(E) ~\ \text{(respectively,~}\, \mu(F)\,\leq\, \mu(E){\rm )}\, .
$$
A semistable vector bundle is called polystable if it is a direct sum of stable
vector bundles.
\end{definition}

\begin{thm} A Hermite-Einstein monopole on $X$ defines a polystable meromorphic 
structure. If the Hermite-Einstein is irreducible then the meromorphic structure is 
stable.
\end{thm}

\begin{proof}
To see this, we make a few remarks. 
\begin{itemize} 
\item A meromorphic subbundle $F$ of $E$ of rank $k$ defines a meromorphic section 
$s_F$ of $Hom(\bigwedge^k F\, , \bigwedge^kE)$, invariant under translation by $t$, and so 
a section of $L\otimes\bigwedge^k E$, with $L$ being the line bundle $\bigwedge^k F^*$.

\item If $E\, ,F$ are of degrees $k\, ,k'$ respectively, and of ranks $n\, ,n'$ respectively, 
the degree of $L\otimes\bigwedge^k E$ is $-k'n+kn'$. It is positive or negative 
depending on whether the difference $\mu(E)-\mu(F)$ of slopes is positive or negative.

\item A Hermite-Einstein monopole structure on $E$ induces a natural Hermite-Einstein 
monopole structure on $L\otimes\bigwedge^k E$

\item A covariant constant section $s_F$ of $L\otimes\bigwedge^k E$, coming from a 
subbundle $F$ of $E$, induces a sub-monopole of $E$.
\end{itemize}
The theorem then follows from the preceding proposition.
\end{proof}

\section{From meromorphic structures to monopoles}

Thus, a Hermite-Einstein monopole on $X$ defines a semistable meromorphic structure. 
By the general Kobayashi-Hitchin correspondence, this should yield a bijective map. The 
main difficulty is in showing that the map is surjective: given a semistable 
meromorphic structure, one would like to find a hermitian structure on the bundle such 
that the result satisfies the Hermite-Einstein condition of Definition 
\ref{HEmonopole}. This amounts to solving a heat equation on the metric. We remark:

(1)\, The case when the bundle $\pi\,:\,X\,\longrightarrow\, \Sigma$ is trivial, 
meaning $\pi$ is the projection of $\Sigma\times S^1$ to $S$, is covered in \cite{ChHu}. 
In this case the corresponding manifold $N$ is K\"ahler, and one can appeal to the basic
theorem of Simpson (\cite{Simpson-Hodge-structures}), and show the existence of a solution
to the equation away from $P$, corresponding to our meromorphic structure. The singularities 
fall into the category covered by Simpson's theorem. One can then appeal to an idea 
developed in the work of Pauly (\cite{pauly}) to show that the result has the right 
Dirac type singularities at $P$.

(2)\, In the case which concerns us, one would have the required theorem if there were 
no singularities. Indeed, on a general closed Gauduchon surface, the Kobayashi-Hitchin 
correspondence has been established by Buchdahl (\cite{Buchdahl}). More generally, 
Jacob (\cite{Jacob}) has proven the more general theorem of Simpson, but again only in 
the case where there are no singularities.

It thus seems likely that the theorem extends to the case that concerns us here. The 
main technical issue seems to be that in the Gauduchon case one does not have the 
Donaldson functional that controls the heat flow near the singularities. We would like 
to thank Adam Jacob for explaining this to us.
 
In any case, one still has the injectivity, as in (\cite{ChHu}): If one has two 
Hermite-Einstein monopoles $E,E'$, such that the corresponding meromorphic structures
$\mathcal E\, , \mathcal E'$ are isomorphic (i.e., through a bundle map on $X\setminus P$
which intertwines the 
holomorphic structures, and preserves the singularity structure -- they therefore have 
the same degree), one has a holomorphic section $s$ of the bundle $\mathcal E^*\otimes 
\mathcal E'$. On the other hand, one has an Hermite-Einstein monopole structure on
$E^*\otimes E'$, with constant $0$, and so the section $s$ must be covariant constant
and commute with $\phi$, and so must define a monopole isomorphism.
 
\section{Holomorphic data on the curve $\Sigma$}\label{se5}

We continue with our assumption that the Sasakian manifold $X$ is regular. So the projection
$\pi$ in \eqref{e1} makes $X$ a principal $S^1$--bundle over the Riemann surface $\Sigma$.

\subsection{Reducing to the curve}

Suppose that we are given a meromorphic structure on a vector
bundle $E$ over the Sasakian three-fold $X$. Let us cover $\Sigma$ by open sets 
$U_\alpha$, and choose sections $$\psi_\alpha\,:\,U_\alpha\,\longrightarrow\,X\, ;$$ we
assume that the images of these sections do not intersect, and that the images do not
contain any $p_i$. We will also assume that enough $\psi_\alpha\,:\,U_\alpha\,
\longrightarrow\, X$ are chosen so that if $p_i\, ,p_j$ lie on the same orbit (so that
$q_i\,=\,q_j$), there is a $\psi_\alpha(q_i)$ lying on the positive path from $p_i$ to
$p_j$. Let
$$
Q \,:=\, \{q_1,\cdots ,q_\ell\}\ \ \text{ and }\ \ \Sigma_0\,:=\, \Sigma\setminus Q\, .
$$

We have holomorphic bundles $E_\alpha\,=\, E_{\psi_\alpha}$ over $U_\alpha$, obtained by 
restricting the holomorphic structure on $E$ to $\psi_\alpha(U_\alpha)$, and 
meromorphic maps (monodromies of $\nabla^c_\xi$ in the $\xi$ direction)
$$G_\alpha\,:\,E_\alpha\,\longrightarrow\, E_\alpha\, ,$$
which are isomorphisms on $\Sigma_0\bigcap U_\alpha$, and have singularities at
$Q\bigcap U_\alpha$. If $p_i$ is alone on its $S^1$ orbit, meaning
$p_i\,=\, \pi^{-1}(q_i)\bigcap P$, then the singularity type at $q_i$ is $\vec{k}_i$. We
also have maps $$\rho_{\beta\alpha}\,:\,E_\alpha\,\longrightarrow\, E_\beta\, ,$$ defined
over $U_\alpha\bigcap U_\beta$, which are obtained by integrating our partial connection
$\nabla^c_\xi$ in the positive direction, from  $U_\alpha$ to $U_\beta$; these are again
meromorphic with polar divisor supported over $Q$, and elsewhere are isomorphisms over
their domains of definition. There is a twisted cocycle condition:
$$\rho_{\alpha\beta}\rho_{\beta\alpha} \,=\, G_\alpha\, .$$
The twist is due to the fact that one is doing one complete turn around the circle
going from $\psi_\alpha( U_\alpha\cap U_\beta)$ to $\psi_\beta( U_\alpha\cap U_\beta)$
to $\psi_\alpha( U_\alpha\cap U_\beta)$, as one is always moving in the positive
direction. In the same way, one has on triple overlaps:
$$
\rho_{\alpha \beta}\rho_{\beta \gamma}\,= \, \rho_{\alpha \gamma} \,~~\ {\mathrm or}\,
~~\   \rho_{\alpha \gamma}G_\gamma
$$
on each component of $U_\alpha\bigcap U_\beta\bigcap U_\gamma$ depending on whether the
images under $\psi_\alpha\, ,\psi_\beta\, , \psi_\gamma$ of the component occur cyclically
as one goes along the orbits of the circle action in $X$, or not. We would like to
understand the set of solutions $E_\alpha\, , G_\alpha\, , \rho_{\alpha,\beta}$ to these
equations, modulo the obvious transformations given by gauge transformations on the $E_i$. We
refer to solutions of these equations as {\it twisted bundle triples} over $\Sigma$.

We first choose some explicit open subsets $U_\alpha$. One can trivialize the circle 
bundle $X\,\longrightarrow\, \Sigma$ over the complement of any point. This reduces us 
to a local geometry near the point which is essentially that of a power $k$ of the Hopf 
fibration. Let us choose a closed disk $D_1$ inside an open disk $D_2$ around a base 
point $p$. Thinking of these disks in the plane, centered at the origin, let $D_i$ be of 
radius $i$, centered on the origin. Set $U_0 \,=\, X\setminus D_1$, and put $U_{ s}\,=\, 
\epsilon$-neighborhood of the angular sector $$\theta\,\in\, (2\pi (s-1)/(k+1)\, , 2\pi 
s/(k+1))\, ,\ \ s\,=\, 1\, ,\cdots\, ,k+1$$ in $D_2$. The open sets $U_0\, , U_{1}\, , \cdots
\, , U_{k+1}$ cover $X$.

One can choose trivializations of  the fibration $X\,\longrightarrow\, \Sigma$ on $D_2\, ,
U_0$ such that the trivialization over $U_0$ is $\exp(\sqrt{-1}k\theta)$ times that on $D_2$.
If one trivializes the bundle over $U_{s} $ by $\exp(2\pi\sqrt{-1}(-s+1/2)/(k+1) )$ times
the trivialization on $D_2$, one obtains trivializations of the bundle satisfying our
requirements: the trivializations on the $U_{s}\, , s\,=\, 1\, ,\cdots\, ,k+1$ are arranged
anti-cyclically in the circle over overlaps, whereas the cyclic order on
$U_0\bigcap U_{s}\bigcap U_{s+1}$ is $s+1\, , 0\, ,  s$. Our cocycle conditions then become:
\begin{align}\label{twist}
\rho_{s, 1 }\rho_{1, s} \,= \,& \rho_{s, t}\rho_{t, s}\,= \,G_{s}\\
\rho_{0, s} \rho_{s, 0 } \, =\, & G_0\nonumber \\
\rho_{ s+2,s+1 } \rho_{ s+1 ,s  } \, =\, & \rho_{ s+2 ,s }\nonumber  \\
\rho_{ s ,0}\rho_{0 ,s+1 }\, =\, & \rho_{ s ,s+1 }\, .  \nonumber 
\end{align}
We then have:
\begin{prop}
The correspondence between meromorphic structures over $X\setminus P$ and twisted 
bundle triples is bijective.\end{prop}

\subsection{Rank one}

We now look at these equations in rank one. In this case the $G_s$ are functions
because they are endomorphisms. Also, since the cocycle equations tell us $G_s$ are 
conjugate, these functions patch together to give a single meromorphic function $G$.

We would like to find one solution to these equations, for a given $G$, using our
explicit cover. On our open set  $D_2$, we suppose that $G$ has neither zero nor
pole, and fix a $(k+1)$-th root  $G^{\frac{1}{k+1}}$ of $G$.  Let us choose a determination
$\log_s(z^{k+1} )$ of $\log (z^{k+1})$ on each $U_s\, ,~ s\,= \, 1\, ,\cdots \, ,k+1$ with
imaginary part going from $0$ to $2\pi$. On the overlap $U_s\bigcap U_{s+1}$, the two
determinations differ by $2\pi\sqrt{-1}$. Let us set
 \begin{align*} \rho_{s
+1, s } \,=\, & G^{\frac{1}{k+1}},\ s\,=\,1,\cdots ,k\\
 T_{0,s}\, =\, & G^{\frac{k\ log_s(z^{k+1} )}{2\pi(k+1)\sqrt{-1}}}\, .
 \end{align*}
One checks that this can be completed to a solution to  the equations (\ref{twist}).

Given one solution, we can find all the others by tensoring with a line bundle on 
$\Sigma$. Explicitly, if $T_{\beta\alpha}$ are transition functions for a line bundle 
over $\Sigma$, one can get from one solution of our twisted line bundle equations to 
another by $\rho_{\beta\alpha}\,\longmapsto\,\rho_{\beta\alpha}T_{\beta\alpha}$.

\begin{prop}
For a given meromorphic function $G$, the family of solutions in rank one 
to the twisted line bundle equations forms a torsor over the Picard group of the 
Riemann surface $\Sigma$.
\end{prop}

We note that the singularities $p_i$ of a monopole are constrained. Indeed, their types 
$k_i$ and their projections $q_i$ to $\Sigma$ determine the divisor $\sum_{i=1}^\ell
k_iq_i$ of the function 
$G$, which is constrained by Abel's theorem, imposing $g \,=\, {\rm genus}(\Sigma)$ complex 
constraints on the divisor. There are also constraints on the angular coordinates 
$\theta_i$ of the points $p_i$ along the orbits. Indeed, we will see, in the Abelian 
case, that these are linked to the Hermite-Einstein constant $C$, and so, fixing $C$
gives us one real constraint.

To see this, let us consider a fixed rank one triple ${\mathcal E}\,=\, (E_\alpha\, , G\, , 
\rho_{\alpha,\beta})$, with singularities at $p_i$ of type $k_i$. Now choose an angular 
coordinate $\theta$ on $X$ near $p_1$, with $$d\theta(\xi)\,=\,1\, , \ \ 
\theta(p_1)\,=\,0\, ,$$ and
define a family ${\mathcal E}_t$ of triples by keeping the same data as ${\mathcal E}$, but 
moving the singular point $p_1$ along its circle orbit in the positive direction to 
$\theta\,=\, t$; call the result $p_1(t)$. Let $z\,:\,U\,\longrightarrow\, D$ be a
coordinate on $U\,\subset\,\Sigma$ with $q_1\,\in\, U$ corresponding to $z\,=\,0$, and
$D$ the unit disk.
We want to consider the difference in the Hermite-Einstein constants (the constant
$C$ in \eqref{eC}) between
${\mathcal E}_t$ and ${\mathcal E}$. This amounts to computing the induced Hermite-Einstein
degree for the triple corresponding to
$Hom({\mathcal E}\, , {\mathcal E}_t)$. The triple corresponding to
$Hom({\mathcal E}\, , {\mathcal E}_t)$ has the property that away from the
orbit through $p_1$, it is canonically identified with the trivial triple $({\mathcal O}\, ,
\bbi\, , \bbi)$. Near $p_i$, one must take two sections
$$\psi_-\,:\,D\,\longrightarrow \,X\, , \ \psi_+\,:\,D\,\longrightarrow\, X\, ,$$ defined
in coordinates by  $\psi_-(z)\,= \,(z\, ,-t/2)$, $\psi_+(z) \,=\, (z\, , t/2)$. The
corresponding maps are $\rho_{-,+}\,= \,z^{-k}$, $\rho_{+-} \,=\, z^{k}$.  Let $S$ denote
the slit $\{z\,=\,0\, ,\,  \theta\,\in\, [0\, ,t]\}$ on $X$, with interior $S^0$ denoting the
slit $\{z\,=\,0$, $\theta\,\in\, (0\, ,t)\}$. Under the correspondence we have with meromorphic
structures on $X$, the triple for $Hom({\mathcal E}\, , {\mathcal E}_t)$ defines a meromorphic bundle
on $X-\{(z,\theta) \,=\, (0,0), (0,t)\}$, trivialized away from $S$, and the holomorphic
transition function to a neighborhood  $V\,\subset\, \pi^{-1}(D)$ of $S^0$ is  $z^{-k}$.
Following the usual recipe, let $\tau$ be a bump function on $X$, equal to one outside $V$
and equal to zero on a neighborhood of $S^0$ of radius $1/2$. One has a Hermitian metric
on our bundle given in the trivialization on $X\setminus S$ by 
$h\,=\, \tau + (1-\tau) z\overline z^k$,
inducing a Chern connection with curvature component $F_\Sigma(t)$ which has the
expression
$$F_\Sigma(t)\,=\, \frac{\sqrt{-1}}{2}\mu(z,\overline{z})\partial_z 
\partial_{\overline z} (\log ( \tau + (1-\tau) z\overline z^k)\, .$$
If one now computes the invariant 
$$ -2\sqrt{-1} Vol(X) C(t)\,=\,  \int_X F(t) \wedge \alpha\,=\,
\int_V F_\Sigma(t)d\alpha\wedge \alpha \,=\,\int_V F_\Sigma(t)d\alpha\wedge d\theta\, ,$$
and compares $F_{t'}$ with $F_t$, one gets 
\begin{align*}
-2\sqrt{-1} Vol(X) (C(t')-C(t))\,=\,& \int_t^{t'} \int_D \partial_z 
\partial_{\overline z} (\log ( \tau + (1-\tau) z\overline z^k)dz\wedge d\overline z\\
\,=\,& \int_t^{t'}\int_{z\overline z = 1/4} \partial_z  (\log ( z\overline z^k)dz
\,=\, k(t'-t)\, .
\end{align*}
Thus, $C(t)-C(0)\,= \,\frac{k\sqrt{-1} t}{ Vol(X)}$.

We have thus seen that there are $g \,=\, {\rm genus}(\Sigma)$ parameters worth of different meromorphic 
structures for a fixed choice of $p_i$, multiplicities $k_i$, and
Hermite-Einstein degree $C$, if 
there exists one. On the other hand, for these to exist, there are $g$ complex 
constraints on $q_i\,=\, \pi(p_i)$ along the curve $\Sigma$ and one real constraint for the 
location of $p_i$ along the circle orbits (for fixed $C$).

\subsection{Higher rank}

More generally, if one is dealing with vector bundles of higher rank $n$, one can take the 
determinant bundle and compute as above. Thus, if $tr(\vec{k}_i)\,=\, \sum_j 
\vec{k}_{i,j}$, we have the following:

\begin{prop} Let $$t\,=\,(t_1\, ,\cdots\, ,t_\ell)\, .$$ Let $(E_\alpha\, , G_\alpha\, ,
\rho_{\alpha,\beta})_t$ be obtained from
$(E_\alpha\, , G_\alpha\, , \rho_{\alpha,\beta})_0$ by shifting the corresponding
singularities $p_i$ along their circle orbits by $t_i$. Then the Hermite-Einstein degree
$C(t)$ of $(E_\alpha\, , G_\alpha\, , \rho_{\alpha,\beta})_t$ is obtained from the
Hermite-Einstein degree  $C(0)$ of $(E_\alpha\, , G_\alpha\, , \rho_{\alpha,\beta})_0$ by 
$$C(t)\, =\, C(0) + \frac{\sqrt{-1}}{ n Vol(X)} \sum_{i=1}^\ell tr(\vec{k}_i)t_i\, .$$
\end{prop}

To understand our parameter space in higher rank, of course, things are not so simple: 
matrices do not always commute. One can Abelianize the problem, however, by what is now 
a classical construction: passing to a spectral curve \cite{Hitchin-construction, 
Markman}.

We have noted that the endomorphisms $G_\alpha$ are all conjugate to each other. This
means that there is an invariant spectral curve $S$, cut out in $\Sigma\times{\mathbb P}^1$
by the equations
$$\det (G_\alpha (z)-\eta\bbi) \,=\, 0\, .$$
Moreover, over each $U_\alpha $, we have quotient sheaves $\mathcal L_\alpha$ supported over
the spectral curve in $U_\alpha \times {\mathbb P}^1$. Let $\sigma\,:\,
\Sigma \times \bbp^1\,\longrightarrow \,\Sigma$ be the projection. Consider
the short exact sequence
$$0\,\longrightarrow \,\sigma^*E_\alpha \otimes {\mathcal O}(-1) \,\buildrel{G_\alpha (z)-
\eta\bbi}\over{\longrightarrow}\, \sigma^*E_\alpha\,\longrightarrow\,
 {\mathcal L_\alpha } \,\longrightarrow\, 0\, .$$
This encodes the pair $(E_\alpha\, , G_\alpha )$, where $E_\alpha\,= \,\sigma_*
{\mathcal L_\alpha }$ and $G_\alpha \,= \,\sigma_*(\times \eta)$, in other words,
$G_\alpha$ is multiplication by
the fiber coordinate. Now let us consider overlaps: on $U_\alpha \bigcap U_\beta$, we have
a diagram
$$ \begin{matrix}
0&\longrightarrow &\sigma^*E_\alpha \otimes {\mathcal O}(-1) &
\buildrel{G_\alpha (z)-\eta\bbi}\over{\longrightarrow} &\sigma^*E_\alpha &\longrightarrow
&{\mathcal L_\alpha } &\longrightarrow& 0\\
&&~\,~\, \Big\downarrow \rho_{\beta\alpha}&&~\,~\,~\,\Big\downarrow\rho_{\beta\alpha}&&
~\,~\,~\, ~\,\Big\downarrow\rho'_{\beta\alpha}\\
 0&\longrightarrow &\sigma^*E_\beta \otimes {\mathcal O}(-1)
&\buildrel{G_\beta (z)-\eta\bbi}\over{\longrightarrow} &\sigma^*E_\beta 
&\longrightarrow &{\mathcal L_\beta }
&\longrightarrow& 0\end{matrix}$$
 On triple overlaps, one gets $ \rho'_{\gamma\beta}\rho'_{\beta\alpha}\,=\,
\rho'_{\gamma \alpha}$ or $\rho'_{\gamma \alpha} \eta$, depending on whether the images of
the open subsets $U_\alpha\, , U_\beta\, , U_\gamma$ are arranged cyclically or not in
$X$. Now, if we suppose that the curve $S$ is smooth, reduced, then $\mathcal L_\alpha$ 
will be line bundles. We thus have obtained a twisted line bundle over the spectral 
curve, and so we have the following:
 
\begin{prop}
Fixing the spectral curve, the family of twisted vector bundles is a torsor over the 
Picard variety of the spectral curve.
\end{prop}

Of course, here, if we want Hermite-Einstein monopoles, one must worry about stability. 
One advantage of the spectral curve approach is that if the spectral curve is 
irreducible and reduced, there are no subobjects, as there are no sub-spectral curves.

\subsection{Gerbe-like structure}

We close this section with the comment that the data in the meromorphic bundle 
structure on $X$ induces a structure which rather resembles Murray's bundle gerbes. 
Indeed, if $X^{[2]}\,\longrightarrow\, \Sigma$ is the fiber product of $X$ with itself, we have 
a $\bbz$-fold cover $\widetilde X^{[2]}$ of $X^{[2]}$, given as pairs of points $x\, ,y$ 
in the fiber over $X$ plus a homotopy class of paths from $x$ to $y$ along the fiber. 
(The inverse image of a point in $\Sigma$ would thus be $S^1\times \bbr$.) Given a 
bundle $E$ on $X$, there is a natural bundle $Hom_E$ on $X^{[2]}$ and by lifting on 
$\widetilde X^{[2]}$, given over $(x\, ,y)$ by $Hom(E_x,\, E_y)$. This has natural maps 
$Hom(E_x,\, E_y)\otimes Hom(E_y,\, E_z)\,\longrightarrow\, Hom(E_x,\, E_z)$, and this is
one of the essential properties of a bundle gerbe, defined by Murray in the rank one case. Our 
remark is that in our case we have a natural section of $Hom_E$, given by our 
integrating $\nabla^c$ over $\widetilde X^{[2]}$ along the fibers, and this section 
respects the multiplication, so that $s(x,y) \times s(y,z)\,= \,s(x,z)$.
 
For our meromorphic bundles, parallel transport by $\nabla^c_\xi(s)\,=\,0$ along the $S^1$ 
preserves the eigenspaces of the holonomy, and so the kernel of the difference of the 
monodromy and any multiple $\eta \bbi$ of the identity map. If $\pi\,:\, S\,\longrightarrow
\,\Sigma$ is the spectral curve, taking the fiber product $Y\,=\, S\times_\Sigma X\,\subset 
\,\bbp^1\times X$. There is a well defined line bundle $\mathcal L$, and so another line 
bundle $Hom_{\mathcal L}$ along $Y^{[2]}$, again equipped with a natural section when 
one lifts to the $\bbz$-cover $\widetilde Y^{[2]}$.

\section{Equivariant bundles on regular Sasakians}

In this section we will reduce the study of holomorphic vector bundles on quasiregular
Sasakians to that of holomorphic vector bundles on regular Sasakians.

Let $X$ be a quasiregular Sasakian threefold. The map $\pi$ in \eqref{e1} fails to be a
submersion outside  finitely many  points of $\Sigma$. Let $x_1\, ,\cdots\, ,x_m$ be the
points of $\Sigma$ such that the complement
$$
\Sigma_0\,:=\, \Sigma\setminus \{x_1\, ,\cdots\, ,x_m\}
$$
satisfies the condition that the restriction
$$
\pi\vert_{\pi^{-1}(\Sigma_0)}\, :\,\pi^{-1}(\Sigma_0)\,\longrightarrow\,\Sigma_0
$$
is a submersion.
For any $z\, \in\, \pi^{-1}(x_i)$, $1\,\leq\, i\, \leq\, m$, the isotropy of $z$ for the
action of $S^1$ on $X$ constructed using the Reeb vector field is a nontrivial finite
cyclic group. Let $\nu_i$ be the order the isotropy subgroup of $z\,\in\,
\pi^{-1}(x_i)$.

We assume the following:

At least one of the following four conditions hold:
\begin{enumerate}
\item $\text{genus}(\Sigma)\, \geq\, 1$,

\item $m\, \geq\, 3$,

\item if $\text{genus}(\Sigma)\, =\, 0$ and $m\,=\, 2$, then $\nu_1\,=\, \nu_2$, and

\item $\text{genus}(\Sigma)\, =\, 0$ and $m\,=\, 0$.
\end{enumerate}

We now recall a theorem of Bundgaard--Nielsen--Fox; see \cite[p. 29, Theorem 1.2.15]{Nam}
and \cite[p. 26, Proposition 1.2.12]{Nam}.

\begin{thm}[Bundgaard--Nielsen--Fox]\label{BNF}
There is a finite Galois covering
$$
\delta\, :\, \widetilde{\Sigma}\, \longrightarrow\, \Sigma
$$
such that
\begin{itemize}
\item $\delta$ is unramified over the complement $\Sigma_0$,

\item for any $1\, \leq\, i\, \leq\, m$, the isotropy of any $y\, \in\, \delta^{-1}(x_m)$
is a cyclic group of order $\nu_i$.
\end{itemize}
\end{thm}

Let $\Gamma\,:=\, \text{Gal}(\delta)$ be the Galois group of the covering $\delta$.

Consider the fiber product $\widetilde{\Sigma}\times_\Sigma X$. Although it is singular,
it has a natural desingularization $\widetilde{X}$ such that the natural projection
$$
\widetilde{\delta}\, :\, \widetilde{X}\, \longrightarrow\, X
$$
is a ramified Galois covering with Galois group $\Gamma$, and the natural projection
$$
\widetilde{\pi}\, :\, \widetilde{X}\, \longrightarrow\, \widetilde{\Sigma}
$$
defines a principal $S^1$--bundle. Therefore, $\widetilde{X}$ is a
regular Sasakian manifold.

Holomorphic vector bundles on $X$ are identified with $\Gamma$--equivariant holomorphic
vector bundles on $X$. This correspondence preserves semistability and polystability,
because equivariant polystability (respectively, equivariant semistability) coincides
with usual polystability (respectively, semistability). Also, all our constructions of twisted
bundle triples go over; one simply then has a condition of equivariance that gets added to the mix.

\end{document}